\documentclass[conference]{IEEEtran}
\usepackage{amssymb,amsthm}
\theoremstyle{plain}

\newtheorem{alg}{Algorithm}
\newtheorem{proposition}[alg]{Proposition}

\theoremstyle{definition}
\newtheorem{definition}[alg]{Definition}

\ifCLASSINFOpdf
\else
\fi
\begin{document}
%
\title{An Algebraic View to Gradient Descent Decoding}

\author{\IEEEauthorblockN{M. Borges Quintana and M.A. Borges Trenard}
\IEEEauthorblockA{Facultad de Matem\'atica y Computaci\'on\\
Universidad de Oriente\\
Santiago de Cuba, Cuba\\
{\tt mijail@mbq.uo.edu.cu \, mborges@mabt.uo.edu.cu}}
\and

\and
\IEEEauthorblockN{I. M\'arquez-Corbella and E. Mart\'{\i}nez-Moro}
\IEEEauthorblockA{SINGACOM group\\
Universidad de Valladolid\\
Castilla, Spain\\
{\tt http://www.singacom.uva.es}\\
{\tt imarquez@agt.uva.es\, edgar@maf.uva.es}
}}


%


\maketitle

\begin{abstract}
There are  two gradient descent decoding procedures for binary codes proposed independently by Liebler and by   Ashikhmin and Barg.  Liebler  in his paper \cite{Liebler}  mentions that both algorithms have the same philosophy but in fact they are rather different. The purpose of this communication is to show that both algorithms can be seen as two ways of understanding the reduction process algebraic monoid structure related to the code.  The main tool used  for showing this is the Gr\"obner representation of the monoid associated to the linear code. 
\end{abstract}


%
\IEEEpeerreviewmaketitle

\section{Introduction}
From now on  a code $\mathcal C$ will be a binary linear code of length $n$ and dimension $k$, i.e. a $k$-dimensional linear subspace of $\mathbb F_2^n$ where $\mathbb F_2$ is the field of two elements. Let $\mathrm{d}(\cdot,\cdot)$, $\mathrm{wt}(\cdot)$ be the Hamming distance and the Hamming weight on $\mathbb F_2^n$ respectively. Let $d$ denote the minimal  Hamming distance of the code $\mathcal C$. 

Given a code $\mathcal C$  and let $\mathbf r$ be a received word in $\mathbb F_2^n$ the \emph{complete decoding problem} (CDP) addresses to determine a codeword  $\mathbf c\in\mathcal C$ that is closest to $\mathbf r$. The \emph{$t$-bounded  distance decoding problem} ($t$-BDP) is to determine a codeword  $\mathbf c\in\mathcal C$ such that   $d(\mathbf r,\mathbf c)\leq t$ if such codeword exists. If $t=\lfloor\frac{(d-1)}{2}\rfloor$ then the solution of the $t$-BDP is unique and if $t=\rho$ the covering radius then $t$-BDP is the same as CDP. Both problems are quite related to the \emph{coset weights problem} ($t$-CWP) that can be stated as follows, given a binary  $r\times n$  matrix and an $r$-dimensional vector $\mathbf s$ and $t\in \mathbb Z_{\geq 0}$, does a binary vector $\mathbf e\in \mathbb F_2^n$ exist such that $w(\mathbf e)\leq t$ and $H\mathbf e=\mathbf s$?
All these problems  have been shown to be  NP-complete \cite{Barg,Berle} even if preprocessing is allowed    \cite{BN}. 

Recently complete decoding and particularly \emph{gradient descent complete decoding} have gain new interest  related to the decoding of LDPC codes, in fact Liebler in \cite{Liebler} says that there is not a clear answer to the question of which parameters of a code could help to recognize and to implement a gradient descent decoding function for the code having the the coset leaders as output. Moreover in the same paper the author makes a distinction between two gradient descent decoding  algorithms (GDDA)  that we will denote by leader GDDA (l-GDDA) and test-set GDDA (ts-GDDA) that are claimed to be different (see section \ref{sec:GDDA} for  formal definitions of the algorithms).

The purpose of this work is to show that both algorithms can be seen as two ways of understanding the reduction process within  algebraic monoid structure related to the code. For that aim the main tool used will  the Gr\"obner representation of the monoid associated to the linear code \cite{BBM}. The structure of the paper will be as follows. Section \ref{sec:GDDA} will show the two gradient descent decoding algoritms, Section \ref{sec:GR} will give a brief review to the Gr\"obner representation of a code and its associate structures. Secion \ref{sec:Reduction} will show the main result, i.e. how the two GDD algorithms can be seen as reduction associated to the  Gr\"obner representation of the code.

\section{Gradient descent decoding  algorithms}\label{sec:GDDA}

In this section we will briefly describe two gradient descent decoding  algorithms. The first one will be the  \emph{leader GDDA} and can be stated as follows. Let us denote by $\overline{\mathbf y}$ the coset in $\mathbb F_ 2^n/\mathcal C$ containing ${\mathbf y}$  and $\mathrm{wt}(\overline{\mathbf y})$ the weight of one of its leaders.

\begin{alg}\label{alg:l} l-GDDA
\begin{enumerate}
 \item[] \textbf{Input:} $\mathbf r$ the received word.
 \item[] \textbf{Output:} A codeword  $\mathbf c\in\mathcal C$ that is closest to $\mathbf r$
 \item[] Repeat until $\mathrm{wt}(\overline{\mathbf r})=0$
 \begin{enumerate}
                                                  \item Compute $\mathbf r^\prime$ such that $\mathrm{wt}(\mathbf r-\mathbf r^\prime)=1 $ 
                                                  and \\ $\mathrm{wt}(\overline{\mathbf r}) \geq \mathrm{wt}(\overline{\mathbf r^\prime}) $
                                                 
                                               \item
                                                  $\mathbf r \leftarrow \mathbf r^\prime$
                                                 \end{enumerate}
\item[] \textbf{Return} $\mathbf c=\mathbf r$.

\end{enumerate}
\end{alg}

Note that in each of the steps of the algorithm the vector $\mathbf r$ changes between different cosets of 
$\mathbb F_ 2^n/\mathcal C$  until it arrives to the $ \overline{\mathbf 0}$ coset, i.e. the code itself.
This is essentially the same as syndrome decoding broken up in smaller steps. The paper \cite{Liebler} presents the first such
construction method of  a gradient function $\gamma: \mathbb F_ 2^n/\mathcal C\rightarrow \mathbb Z$ such that is a strictly increasing function of $\mathrm{wt}(\overline{\mathbf m})$ for performing such a  l-GDDA. 

For understanding the next GDD algorithm we will need some knowledge of minimal (support) codewords. The \emph{support of a codeword} $\mathbf c\in \mathcal C$ will be the set of its non-zero positions, i.e.
$\mathrm{supp}(\mathbf c)=\{i\mid \mathbf c_i\neq 0\}$.
\begin{definition}
 A codeword $\mathbf m$ in the code $\mathcal C$ is said to be minimal if there is no other codeword $\mathbf c\in \mathcal C$ such that 
$$\mathrm{supp}(\mathbf c)\subseteq\mathrm{supp}(\mathbf m).$$
\end{definition}
 We will denote by $\mathcal M_\mathcal C$ the set of all the minimal codewords of $\mathcal C$.  The usual way of defining a steepest descent method in the Hamming space is to construct a  \emph{test set} $\mathcal T\subseteq \mathbb F_2^n$. A test set is a set of codewords such that every word $\mathbf y$ either lies in $V(\mathbf 0)$  (the Voronoy region of the all-zero vector) or there is a $\mathbf t\in \mathcal T$ such that $\mathrm{wt}(\mathbf y-\mathbf t)< \mathrm{wt}(\mathbf y).$ The gradient-like or 
 \emph{test set GDDA} is stated as follows (see \cite{Barg} for further details and correctness of the algorithm)
\begin{alg}\label{alg:ts} ts-GDDA
\begin{enumerate}
 \item[] \textbf{Input:} $\mathbf r$ the received word.
 \item[] \textbf{Output:} A codeword  $\mathbf c\in\mathcal C$ that is closest to $\mathbf r$
\item[] $\mathbf c\leftarrow\mathbf 0$ 
\item[] Repeat until  no $\mathbf t\in\mathcal T$ is found such that
 $$\mathrm{wt}(\mathbf r- \mathbf t)< \mathrm{wt}(\mathbf r)$$
 \begin{enumerate}
                                                  \item $\mathbf c \leftarrow\mathbf c+\mathbf t$
                                                 
                                               \item
                                                  $\mathbf r \leftarrow \mathbf r-\mathbf t$
                                                 \end{enumerate}
\item[] \textbf{Return} $\mathbf c$.

\end{enumerate}
\end{alg}

It is pointed in \cite{Barg} that setting $\mathcal T=\mathcal M_\mathcal C$ in the previous Algorithm~\ref{alg:ts} the so call minimal vector algorithm performs  complete minimum distance decoding.

\section{Gr\"obner representation and related structures}\label{sec:GR}

In this section we will show some basic results on the Gr\"obner representation of a code $\mathcal C$. In fact it is related to the additive representation of the  monoid $\mathbb F_2^n/\mathcal C$. We will try to keep this section as Gr\"obner basis technology-free as possible.  For some references on Gr\"obner representations of codes  and its implementations see \cite{BBM, bbm2, bbm3, bbfm, GBLA}. Let $\mathbf e_i\in \mathbb F_2^n$ be the vector with all its entries $0$ but a $1$ in the $i$th-position.

\begin{definition}\label{def:gr} 
 A \emph{Gr\"obner representation} of  $\mathbb F_2^n/\mathcal C$ is a pair $N,\phi$ where $N$ is a transversal of the cosets in  $\mathbb F_2^n/\mathcal C$ (i.e. one element of each coset)  such that $\mathbf 0\in N$ and for each $\mathbf n\in N\setminus \{\mathbf 0 \}$ there exists a $\mathbf e_i$, $i\in \{1,2,\ldots , n\}$ such that $\mathbf n=\mathbf n^\prime + \mathbf e_i$ with $\mathbf n^\prime\in N$ and a mapping $$\phi: N\times  \{\mathbf e_1,\mathbf e_2,\ldots , \mathbf e_n\}\rightarrow N$$ such that the image of the pair $(\mathbf n, \mathbf e_i)$ is the element representing the coset that contains $\mathbf n+ \mathbf e_i$.
\end{definition}

The  word Gr\"obner is not casual as we will see it with the following construction. Let us consider the binomial ideal 
\begin{equation}
 \mathcal I_{\mathcal C}=\left\langle \left\lbrace \mathbf x^{\blacktriangle\mathbf w_1}-\mathbf x^{\blacktriangle\mathbf w_2}\mid \mathbf w_1-\mathbf w_2\in \mathcal C   \right\rbrace\right\rangle \subseteq {\mathbb K}[x_1,\ldots , x_n]
\end{equation}
where the characteristic crossing function $\blacktriangle:  \mathbb Z_2^n \rightarrow \mathbb Z^n$ replaces the class of  $0,1$ by the same symbols regarded as
integers; ${\mathbb K}$ is an arbitrary field and  if  $\blacktriangle\mathbf w=(w_1,\ldots, w_n)$ then  $ \mathbf x^{ \blacktriangle\mathbf w}=\prod x_i^{w_i}$. If we consider a degree compatible ordering $\prec$ and we compute a Gr\"obner basis $\mathcal G_\prec$ w.r.t.  $\prec$ of the ideal $ \mathcal I_{\mathcal C}$  the normal form of any monomial  $\prod x_i^{w_i}$  corresponds with the syndrome of  the word $\blacktriangledown   (w_1,\ldots, w_n)$ where the map  $\blacktriangledown$ is reduction modulo $2$. Thus we can take $N$ in Definition~\ref{def:gr} as the vectors  $\blacktriangledown   (w_1,\ldots, w_n)$ such that $\prod x_i^{w_i}$ is a normal form w.r.t. $\mathcal G_\prec$, i.e. the syndromes of the code. Note also that $\phi$ is just  given by the  multiplication tables of the normal forms times the variables $x_i$ in the ring  ${\mathbb K}[x_1,\ldots , x_n]/\mathcal I_{\mathcal C}$. This is standar way of representing the quotient by  an ideal $\mathcal I_{\mathcal C}$ using the FGLM  algorithm (see \cite{Mora} chapter 29  for a complete reference on Gr\"obner basis topics).  Moreover, the Gr\"obner representation of a code can be computed with a modification of the   FGLM algorithm \cite{bbm2}, one implementation in  GAP \cite{GAP} of this algorithm can be found in \cite{GBLA}.  

The binomial ideal  $\mathcal I_{\mathcal C}$ can be seen also as a kernel of a modular integer  linear program problem stated as follows.
Let $H\in \mathbb Z^{m\times n}$  be a $m\times n$ matrix such that $\blacktriangledown H $ is a parity check matrix of $\mathcal C$ and $\mathbf b\in \mathbb Z^m$.
  \begin{equation}\label{eq:ip}
 IP_{H}(\mathbf b)\equiv\left\lbrace
 \begin{array}{l}
  \min\left\lbrace {(1,1,\ldots,1)\cdot \mathbf u^t}\right\rbrace \\
\mathbf u\in \mathbb Z_{\geq 0}\\
H\cdot \mathbf u^t=\mathbf b\, \mathrm{mod} \,2.
 \end{array}
 \right. 
\end{equation}
 Ikegami and Kaji \cite{IK} studied the kernel of this problem  related with the maximum likelihood decoding problem. It has been also studied in \cite{MM} in order to describe the combinatorics of the minimal codewords of the code $\mathcal C$.
 
Associated to the Gr\"obner representation we can define the \emph{border of a code} \cite{Border} as follows

\begin{definition} Let  $\mathcal C$ be a code and  $H$  a parity check matrix of $\mathcal C$, let  $(N,\phi)$ be a {Gr\"obner representation} of  $\mathbb F_2^n/\mathcal C$. Then the  \emph{ border of the code } $\mathcal C$ w.r.t. $(N,\phi)$  is the set 
\begin{equation}\label{def:bor1}
 \begin{array}{rl}
B(\mathcal C)= &\left\lbrace (\mathbf n_1+\mathbf e_i,\mathbf n_2)\mid i\in \{1,\ldots n\}, \mathbf n_1+\mathbf e_i\neq\mathbf n_2,\right.  \\ 
& \quad \left. \mathbf n_1,\mathbf n_2\in N \hbox{ and }  H\cdot (\mathbf n_1+\mathbf e_i)=H\cdot \mathbf n_2\right\rbrace , 
\end{array}
\end{equation}
\end{definition}

An important remark is that both components of an element in the set $B(\mathcal C)$ are in the same coset, i.e. their   sum is in the code. We can also describe the border in  as
\begin{equation}\label{def:bor2}
B(\mathcal C)= \left\lbrace \left( \mathbf n+\mathbf e_i,\phi(\mathbf n,\mathbf e_i)\right) \mid i\in \{1,\ldots n\},  \mathbf n\in N \right\rbrace \setminus \{(\mathbf x, \mathbf x) \}.
\end{equation}

The border of a code $B(\mathcal C)$  is associate to the  border basis of the ideal $\mathcal I_\mathcal C$. The conection between the ideal $\mathcal G_\prec$ comes from the well known fact that 
that every Gr\"obner basis with respect to a degree-compatible term ordering can be
extended to a border basis (see \cite[p. 281ff]{Border2}) but not every border basis is an extension of a Gr\"obner
basis.   The preference of border bases over Gr\"obner bases in our case arises from the iterative generation of
linear syzygies, inherent in the  linear algebra algorithm used in \cite{bbfm}, which allows for successively approximating the
basis degree-by-degree, i.e. weight-by-weight.

\section{Gradient descent decoding and reduction}\label{sec:Reduction}

Given a code  $\mathcal C$  and its corresponding  Gr\"obner representation $(N,\phi)$ we can accomplish two types of reduction that we will see are associated to Algorithms~\ref{alg:l},~\ref{alg:ts} above. Thus both algorithms obey to  the same algebraic structure. 

\subsection{Reduction by $\phi$} We shall define the reduction of an element   $\mathbf n\in N$ w.r.t. $\mathbf e_i$ as the element   $\mathbf n^\prime=\phi (\mathbf n, \mathbf e_i)$  and we will denote it by   $\mathbf n\rightarrow_i \mathbf n^\prime$.  For each  $\mathbf y\in \mathbb F_2^n$, $\mathbf y= \mathbf 0+\sum_j \mathbf e_{i_j}$ for some $i_j\in \{1,\ldots n\}$, thus we can iterate a finite number of reductions to find the representative of the coset $\overline{\mathbf y}$  containing $\mathbf y$. Note that in the case that we use $\prec$ defined above the representatives of the classes corresponds with  coset leaders,  we will consider that this is the case from now on. This gives us the following gradient descent decoding algorithm. 

\begin{alg}\label{alg:phi} $(N,\phi)$-reduction GDDA
\begin{enumerate} 
 
 \item[] \textbf{Input:} $\mathbf r$ the received word.
 \item[] \textbf{Output:} A codeword  $\mathbf c\in\mathcal C$ that is closest to $\mathbf r$\\[0.05em]
\item[] \textbf{Forward step} 
\item[] $\mathbf r=\sum_{j=1}^s\mathbf e_{i_j}$. Compute $\mathbf n\in N$ corresponding to the coset  $\overline{\mathbf r}$, i.e.
\begin{enumerate}
\item $\mathbf n=\mathbf 0$.
 \item For $j=1,...,s$ do $$\mathbf n \rightarrow_{i_j}\mathbf n^\prime, \quad \mathbf n \leftarrow \mathbf n^\prime$$
\end{enumerate}
\item[] \textbf{Backward step} 

\item[] While $\mathbf n\neq \mathbf 0$
 \begin{enumerate}
                                                  \item Compute  $\mathbf r^\prime$ such that $\mathbf r^\prime=\mathbf r +\mathbf e_{i_j}$ and $$\mathrm{wt}(\mathbf n) \geq \mathrm{wt}(\phi(\mathbf n, \mathbf e_{i_j})) $$
                                                  \item $\mathbf r \leftarrow \mathbf r^\prime$,$\quad \mathbf n \leftarrow \phi(\mathbf n, \mathbf e_{i_j})$.
                                                 \end{enumerate}
 \item[] {\textbf{Return:}} $\mathbf c=\mathbf r$.
\end{enumerate}
\end{alg}

Note that the previous algorithm is somehow redundant, since at the end of the forward step we end with the coset leader $\mathbf n$ of the class $\overline{\mathbf r}$, thus we can decode without performing the forward step. Anyway we have staded  this way to see the resemblance with Algorithm~\ref{alg:l}. We can modify our  algorithm capturing the needed information of the  {Gr\"obner representation}  as follows.

\begin{definition}\label{def:gr+} 
 Let $(N,\phi)$ \emph{Gr\"obner representation} of  $\mathbb F_2^n/\mathcal C$ and $\{ \mathbf n_i\}_{i=1}^{2^{n-k}}$ an ordering on $N$ with $\mathbf n_1=\mathbf 0$.  We will denote by $(N^\star,\phi^\star)$ the following pair.
$$N^\star =\{ (i,w_i)\in\mathbb Z_{\geq 0}^2\mid w_i=\mathrm{wt(\mathbf n_i)},\, i=1,\ldots , 2^{n-k}\}$$
$$ \begin{array}{rccc}
\phi^\star: & N^\star\times \{\mathbf e_1,\mathbf e_2,\ldots , \mathbf e_n\} &\rightarrow  & N^\star\\
 & (i,w_i) & & \phi^\star ((i,w_i),\mathbf e_j)
\end{array}
$$
such that $\phi^\star ((i,w_i),\mathbf e_j)= (i_j,w_{i_j})$ if $\mathbf n_{i_j}= \phi(\mathbf n_i,\mathbf e_j)$ and $w_{i_j}=\mathrm{wt}(\mathbf n_{i_j})$.
\end{definition}
 
In other words, we keep track only on the ordering of the normal forms representing each coset and the weight of one of its leaders. Note that  $(N^\star,\phi^\star)$ can be easily computed from a Gr\"obner representation $(N,\phi)$  w.r.t. a degree compatible ordering $\prec$ since for $\prec$ the normal forms are coset leaders. Moreover, the way of computing a Gr\"obner representation by FGLM techniques gives us an incremental construction of $N^\star$ ordered non-decreasingly on the second component (see \cite{BBM} for further details), i.e.
$$(i,w_i),(j,w_j)\in N^\star \hbox{ and } i<j \Rightarrow  w_i\leq w_j.$$
Now it is clear that we can decode using only  $(N^\star,\phi^\star)$, thus we can avoid storing the normal forms in the Gr\"obner representation.

\begin{alg}\label{alg:phiestrella} $(N^\star,\phi^\star)$-reduction GDDA
\begin{enumerate} 
 
 \item[] \textbf{Input:} $\mathbf r$ the received word.
 \item[] \textbf{Output:} A codeword  $\mathbf c\in\mathcal C$ that is closest to $\mathbf r$\\[0.05em]
\item[] \textbf{Forward step} 
\item[] $\mathbf r=\sum_{j=1}^s\mathbf e_{i_j}$. Compute $\ell\in \{1,\ldots , 2^{n-k}\}$ corresponding to the coset  $\overline{\mathbf r}$, i.e.
\begin{enumerate}
\item $i= 1,\, w_1=0$.
 \item For $j=1,...,s$ do $$\phi^\star ((i,w_i),\mathbf e_{i_j})=(i^\prime,w_i^\prime), \quad (i,w_i) \leftarrow (i^\prime,w_i^\prime)$$
\end{enumerate}
\item[] Return $i=\ell$.
\item[] \textbf{Backward step} 

\item[] While $i\neq 1$
 \begin{enumerate}
                                                  \item Compute  $\mathbf r^\prime$ such that $\mathbf r^\prime=\mathbf r +\mathbf e_{i_j}$ and $$w_i \geq w_i^\prime$$  where $w_i^\prime$ is the second component of $\phi^\star ((i,w_i),\mathbf e_{i_j})$
                                                  \item $\mathbf r \leftarrow \mathbf r^\prime$,$\quad (i,w_i) \leftarrow \phi^\star ((i,w_i),\mathbf e_{i_j})$.
                                                 \end{enumerate}
 \item[] {\textbf{Return:}} $\mathbf c=\mathbf r$.
\end{enumerate}
\end{alg}

As an intermediate result, from the forward step we already know if the coset has a correctable leader if $w_\ell\leq t=\lfloor\frac{(d-1)}{2}\rfloor $, in that case the backward step gives us a unique solution, if $\ell> t$ then there could be multilple ways of doing the backtracking step depending on the number of leaders in the $\ell$th-coset. Also this algorithm can be use to answer the $t$-CWP problem.

 Note that the backward step is just the  l-GDDA in Algorithm~\ref{alg:l}. As pointed by Liebler \cite{Liebler} in each step of the backtracking procedure we change of coset till we arrive to the $\overline{\mathbf 0}$ coset.

\subsection{Border reduction}
 
 Now taking into account the information on the border of the code $B(\mathcal C)$, we can make a similar reduction substituting in each step the first component of an element of the border by the second one.  More formally,  let $(\mathbf b_1, \mathbf b_2)=\mathbf b\in B(\mathcal C)$, we define the \emph{head}  and the \emph{tail} of $\mathbf b$  as 
$$\mathrm{head}(\mathbf b)=\mathbf b_1,\,\mathrm{tail}(\mathbf b)= \mathbf b_2\in \mathbb F_2^n.$$ 
As pointed before   $\mathrm{head}(\mathbf b)+\mathrm{tail}(\mathbf b)$ is a codeword of $\mathcal C$ for all $\mathbf b\in B(\mathcal C)$ and by its definition (\ref{def:bor1})  it is clear that the information in the border allows complete decoding. The information in the border is somehow redundant, we can reduce the number of codeword in it needed for decoding.

\begin{definition}
A set $R(\mathcal C)$ is the \emph{reduced border for the code} $\mathcal C$
with respect to the order $\prec$ if $R(\mathcal C)\subseteq  B(\mathcal C)$ and it fulfills the following
conditions:
\begin{enumerate}
\item For each pair $(\mathbf n, \mathbf e_i )$  such that $\mathbf n + \mathbf e_i$ is a head in $B(\mathcal C)$ 
       there exists an element in $R(\mathcal C)$ such that its head is $\mathbf h$ where
       $$\mathrm{supp}(\mathbf h) \subseteq \mathrm{supp}(\mathbf n + \mathbf e_i ).$$
\item  Given two elements in $R(\mathcal C)$ and $\mathbf h_1$ , $\mathbf h_2$ their heads, then we
       have that 
$$\mathrm{supp}(\mathbf h_1)\not \subseteq \mathrm{supp}(\mathbf h_2) \hbox{ and } \mathrm{supp}(\mathbf h_2)\not \subseteq \mathrm{supp}(\mathbf h_1).$$
 
\end{enumerate}
 Thus $R(\mathcal C)$  is the set with smallest cardinal that allows us a gradient-like test set decoding using reductions.
\end{definition}

\begin{proposition}
 Let us consider the set of codewords in $\mathcal C$ given by
\begin{equation}
\mathrm{Min}_{red} (\mathcal C) = \{\mathrm{head}(\mathbf b)+\mathrm{tail}(\mathbf b) \mid \mathbf b \in  R(\mathcal C)\}\subseteq \mathcal C.
\end{equation}
Then $\mathrm{Min}_{red} (\mathcal C)\subseteq \mathcal M_\mathcal C$.
\end{proposition}
\begin{proof}
 Let $\mathrm{head}(\mathbf b)+\mathrm{tail}(\mathbf b)=\mathbf c$ where $\mathbf b \in  R(\mathcal C)$ and  suppose $c\notin  \mathcal M_\mathcal C$,
then there exists $\mathbf c^\prime \in\mathcal C$ such that $\mathrm{supp}(\mathbf c^\prime) \subset \mathrm{supp}(\mathbf c)$. Let $\mathbf c_1$ be a vector such
that $\mathrm{supp}(\mathbf c_1) = \mathrm{supp}(\mathbf c) \cap  \mathrm{supp}(\mathrm{head}(\mathbf b))$, thus $\mathbf c_2 = \mathbf c -\mathbf c_1$ fulfills
$\mathrm{supp}(\mathbf c_2)  \subset  \mathrm{supp}(\mathrm{tail}(\mathbf b))$. Let $\mathbf m$ be the maximum between $\mathbf c_1$ and $\mathbf c_2$,
therefore $\mathrm{supp}(\mathbf m) \subset \mathrm{supp}(\mathbf c)$, and $\mathbf m$ is smaller than $\mathrm{head}(\mathbf b)$ and $\mathrm{tail}(\mathbf b)$
which contradicts the fact that $R(\mathcal C)$ is reduced.

\end{proof}

Therefore the set $\mathrm{Min}_{red} (\mathcal C)$ is a minimal test set w.r.t. the order $\prec$ given by minimal codewords that allow  the ts-GDD algorithm stated in  Algorithm \ref{alg:ts}. It can be also seen as a test set for the modular integer program in Equation (\ref{eq:ip}).

\section*{Conclusions}

We have shown an unified approach via the Gr\"obner presentation of a code to two  gradient descent decoding algorithms, one that the search is done changing the coset representative (l-GDDA) and the one given by descending within the same coset (ts-GDDA) that were claimed to be of different nature.  This two algorithms come from two ways of computing the reduction of a monomial modulo a binomial ideal associated to the code. Unfortunately there are some obstructions for  generalizing this approach in a straightforward way to non binary codes mainly motivated by the non-admissibility of the ordering needed for decoding (see \cite{BBM}). Further research lines of the authors point to generalizing the border basis for the non binary case in order to describe the set of minimal codewords of a code.

\end{document}